\documentclass{cccg22}
\usepackage[T1]{fontenc}
\usepackage[utf8]{inputenc}

\usepackage{amsfonts,amsmath,amssymb}
\usepackage{cite,microtype,graphicx}
\usepackage{enumerate}

\title{Locked and Unlocked Smooth Embeddings of Surfaces}
\author{David Eppstein\thanks{Department of Computer Science, University of California, Irvine. This work was inspired by discussions at the 3rd Virtual Workshop on Computational Geometry, held in March 2022, for which we thank the organizers and participants. Research supported in part by NSF grant CCF-2212129.}}

\date{ }

\begin{document}
\thispagestyle{empty}
\maketitle  

\begin{abstract}
We study the continuous motion of smooth isometric embeddings of a planar surface in three-dimensional Euclidean space, and two related discrete analogues of these embeddings, polygonal embeddings and flat foldings without interior vertices, under continuous changes of the embedding or folding. We show that every star-shaped or spiral-shaped domain is unlocked: a continuous motion unfolds it to a flat embedding. However, disks with two holes can have locked embeddings that are topologically equivalent to a flat embedding but cannot reach a flat embedding by continuous motion.
\end{abstract}

\section{Introduction}

Much research in computational geometry has concerned smooth deformations of a shape that preserve its geometric structure. This includes bloomings, continuous and collision-free unfoldings from polyhedra to flat nets that preserve the shape of each face~\cite{SonAma-TRA-04,BieLubSun-CG-05,MilPak-DCG-08,DemDemHar-GC-11,XiLie-IROS-15,HaoKimLie-SCF-18}, the carpenter's rule problem on continuous collision-free motions that straighten a polygonal chain while preserving segment lengths~\cite{Str-FOCS-00,ConDemRot-DCG-03,Par-TAMS-09}, and the closely related study of continuous rigid motions of single-vertex origami patterns~\cite{StrWhi-JCDCG-05,PanStr-CG-10,AbeCanDem-JoCG-16}. When the carpenter's rule problem is generalized to to more complex linkages or three dimensions it can have \emph{locked} configurations, unable to reach a straightened configuration even though there is no topological obstacle to their reconfiguration~\cite{HopJosWhi-SICOMP-84,BieDemDem-DCG-01,BieDemDem-DAM-02,DemLanOro-CG-03,BalChaDem-WADS-09,ConDemDem-DCG-10}.
Demaine, Devadoss, Mitchell, and O'Rourke studied ``folded states'' of simple planar polygons in 3d, consisting of a surface-distance-preserving mapping to 3D together with a consistent ``local stacking order'' at parts of the polygon that are mapped onto each other. As they show, any folded state can be continuously transformed to any other folded state: the configuration space of these states is connected~\cite{DemMit-CCCG-01,DemDevMit-CCCG-04}.

In this work we examine reconfigurability for three natural restricted forms of folded states:
\begin{itemize}
\item Smooth embeddings into $\mathbb{R}^3$, where the embedded surface is doubly differentiable (having a tangent plane everywhere) without self-contact.
\item Polygonal (piecewise linear) embeddings into $\mathbb{R}^3$ without interior vertices, so all ``fold lines'' where the mapping is not locally linear extend entirely across the surface. There should be finitely many connected linear pieces and no self-contact.
\item Planar folded states (flat foldings) without interior vertices. We again require that the mapping be piecewise linear with finitely many pieces and that the fold lines extend entirely across the surface.
\end{itemize}
At each interior point of a smoothly embedded flat surface in $\mathbb{R}^3$ that is not locally flat, the surface bends along a straight ``bend line'' that continues to the surface's boundary~\cite{FucTab-MO-07}. As an everyday example of this phenomenon, when holding a slice of pizza at its crust, keeping the crust flat allows the slice to droop, but bending it extends rigid bend lines lengthwise through the slice, preventing drooping~\cite{TurGooSen-PIMEC-15}. Our restriction against interior vertices of polygonal embeddings and flat foldings provides a combinatorial model of the same phenomenon. We have studied  flat foldings with this restriction (but not their reconfiguration) in previous work~\cite{Epp-JoCG-19}.


In all three models of bending and folding we allow continuous motions that stay in the same model; in particular, in the polygonal embedding model, folds may ``roll''  along the surface rather than remaining fixed in place.
Our folded states are special cases of the ones previously considered by Demaine et al.~\cite{DemMit-CCCG-01,DemDevMit-CCCG-04}, and we retain their notion of a continuous motion as a mapping from the closed unit interval $[0,1]$ to folded states that is continuous under the sup-norm of the distances of mapped points and (for flat foldings) consistent with respect to the local stacking order. The initial configuration of a motion is the mapping for the parameter value $0$, and the final configuration is the mapping for the parameter value $1$. 
For all three of our restricted models of folded states, we prove the following new results:
\begin{itemize}
\item A compact subset of the plane with a continuous shrinking motion has a connected space of restricted folded states: every folded state can be continuously unfolded to a flat state. These sets are topological disks and include the star-shaped domains.
\item There exist subsets of the plane, topologically equivalent to a disk with two holes, that can be locked: they have embeddings that are topologically equivalent (ambient isotopic) to their flat embedding, but cannot be continuously deformed to become flat.
\end{itemize}

\section{Shapes that can shrink into themselves}

\begin{figure}[t]
\centering\includegraphics[width=0.7\columnwidth]{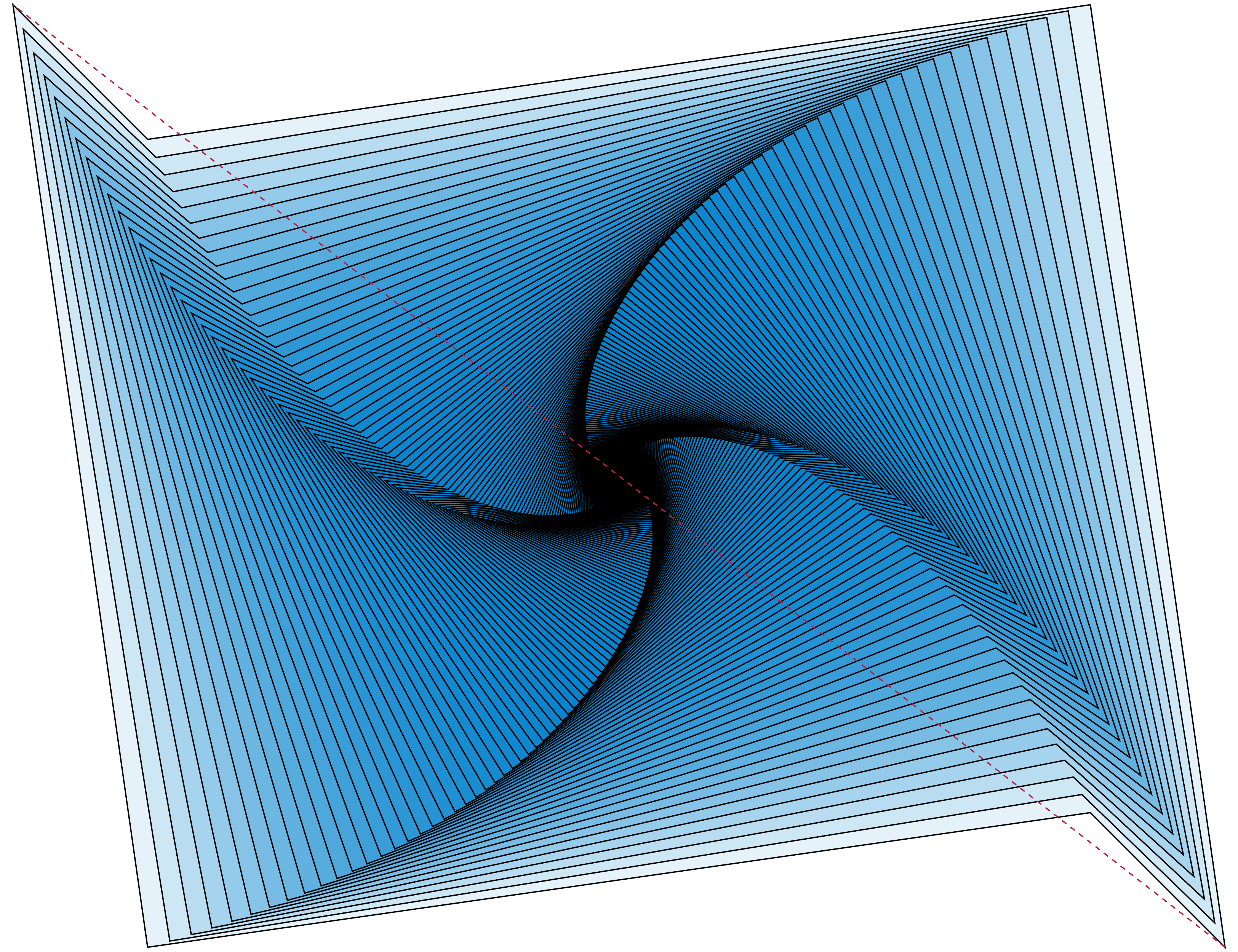}
\caption{A polygon that has a continuous shrinking motion into itself, but that is not star-shaped.}
\label{fig:generalized-star}
\end{figure}

A \emph{star-shaped} polygon, or more generally a star-shaped domain, is a subset $S$ of the plane such that, with an appropriate choice of one point of the plane to be the origin, every scaled copy $pS$ for $p\in[0,1]$ is a subset of $S$ itself. These are widely studied in computational geometry, and can be recognized in linear time~\cite{LeePre-JACM-79}. However, these are not the only shapes that have a continuous shrinking motion of scaled copies of the shape into themselves. \cref{fig:generalized-star} depicts a different type of continuous shrinking motion, in which the shape spirals inwards while shrinking. Such a motion can be described by coordinatizing the plane by complex numbers, again for an appropriately chosen origin (the center of the spiral motion), choosing a complex number $q$ of absolute value less than one, and considering the family of scaled copies $e^pqS$ for $p\in[0,\infty)$. The linear shrinking motion of star-shaped domains is a special case of this type of motion in which $q$ is a positive real number. If any shape $S$ has any continuous shrinking motion of its scaled copies into itself, the start of the motion can be approximated to first order by an inward-spiraling shrinking motion of this type, which can then be continued to a complete inward-spiraling shrinking motion of the same shape. In this sense, this family of continuous shrinking motions is completely general, in the sense that all shapes with a continuous shrinking motion have a continuous shrinking motion of this type, although we will not use this fact. Following Aharonov et al.~\cite{AhaEliSho-JMAA-03}, we call a compact set $S$ that has a continuous shrinking motion of this type a \emph{spiral-shaped domain}.

In an inward-spiraling shrinking motion, each point of the set $S$ traces out a logarithmic spiral, which meets every ray from the center of the motion in a fixed angle $\theta$. The existence of a spiraling motion for a given simple polygon and a fixed choice of $\theta$ can be tested by intersecting polygonal regions, one for each edge, that describe the set of locations for the center where a spiral of this angle would leave the edge towards the interior of the polygon, rather than towards the exterior. This characterization leads to a polynomial-time algorithm for testing the existence of an inward-spiraling motion, by continuously varying $\theta$ over the range $(-\pi,\pi)$ and determining the combinatorial changes in the corresponding intersection of polygonal regions. It is plausible that finding a feasible angle $\theta$ and a center point for that choice of $\theta$ is an LP-type problem of bounded dimension, solvable in linear time, but we have not proved this.

The main results of this section are that every spiral-shaped domain $S$ has connected spaces of smooth embeddings, polygonal embeddings without interior vertices, and flat-foldings without interior vertices. Equivalently, every embedding can be continuously unfolded. The simplest case concerns smooth embeddings.

\begin{theorem}
Every smooth embedding of a spiral-shaped domain has a continuous motion, through smooth embeddings, to a flat embedding.
\end{theorem}

\begin{proof}
Let $S$ be the given spiral-shaped domain, and $f:S\to\mathbb{R}^3$ be its smooth embedding.
Parameterize an inward-spiraling shrinking motion of $S$ as $s_i:S\to S$ where $i\in(0,1]$,  $s_1$ is the identity, and each $s_i$ scales $S$ by a factor of $i$, converging as $i\to 0$ to a single central point of the motion (which may or may not be on the boundary of~$S$).

Our proof converts this parameterized family of scalings to a parameterized family of smooth embeddings of $S$ at a single scale, by composing $s_i$, $f$, and a function that expands $\mathbb{R}^3$ by a factor of $1/i$ to restore the original size of $S$. The obvious expansion function $\mathbb{R}^3$ by $(x,y,z)\mapsto (x/i,y/i,z/i)$ does not work, because of the following issues:
\begin{itemize}
\item Composing $s_i$, $f$, and an expansion by $1/i$ does provide a continuous motion of smooth embeddings on the half-open interval $(0,1]$, whose curvature tends towards zero as the parameter goes to zero. However, we need a continuous motion on the closed interval $[0,1]$ for which the limiting embedding at parameter value zero exists and is completely flat. 

\item When the origin of $\mathbb{R}^3$ does not belong to all of the rescaled and smoothly embedded copies of $S$, the
composition with $\mathbb{R}^3$ by $(x,y,z)\mapsto (x/i,y/i,z/i)$, as $i\to 0$, will produce smooth embeddings of $S$ whose
distance from the origin is inversely proportional to $i$, preventing them from having a limit. We can prevent this by choosing coordinates for $\mathbb{R}^3$ that have as their origin $f(p)$, where $p$ is the limit point of the inward-spiraling shrinking motion on $S$. In this way, the composition of $s_i$, $f$, and an expansion by $1/i$ will act as the identity on this point, and more strongly will preserve the tangent plane of the surface at that point.

\begin{figure}[t]
\centering\includegraphics[width=0.8\columnwidth]{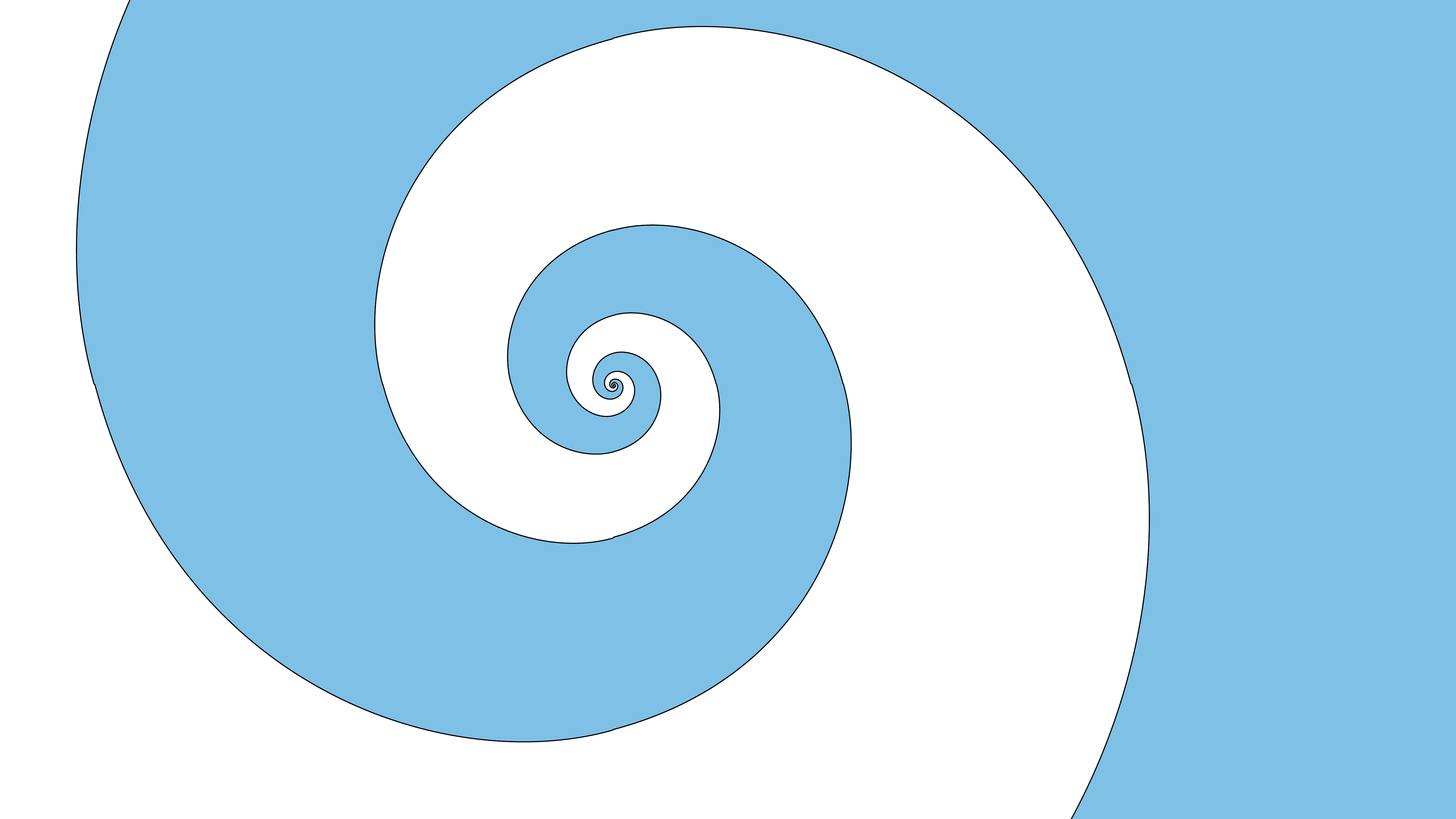}
\caption{A spiral-shaped domain between two logarithmic spirals, whose shrinking motion cannot be shortcut to a linear shrinking motion.}
\label{fig:spiral-domain}
\end{figure}

\item This still does not complete the proof, because a spiral inward-shrinking motion of the domain, composed with $f$ and $(x,y,z)\mapsto (x/i,y/i,z/i)$, will cause the embeddings to rotate at increasing speed as $i\to 0$, preventing the continuous motion from having a flat limiting surface at $i=0$.
For the polygonal domain of \cref{fig:generalized-star}, this problem can be circumvented by switching from the spiral inward-shrinking motion to a linear scaling transformation in $S$, once the scaled copies of $S$ become small enough that this linear scaling stays entirely within $S$. However, for some other shapes, such as the domain between two logarithmic spirals depicted in \cref{fig:spiral-domain}, switching to linear scaling is never possible.

Instead, we address this third issue by choosing a reference vector tangent to $S$ at $f(p)$. We compose $s_i$, $f$, an expansion by $1/i$, and a rotation of $\mathbb{R}^3$ (with axis perpendicular to the tangent plane at $f(p)$)  that restores this vector to its original direction.
\end{itemize}
Let $f_i$ denote the resulting composition of $s_i$, $f$, an expansion of $\mathbb{R}^3$ with a careful choice of origin, and a rotation that restores the original directions of vectors in the tangent plane to the embedded surface.
Then $f_i$, for values of $i$ in the half-open interval $(0,1]$, describes a continuous motion with $f_1=f$ as one endpoint of the motion. The maximum curvature of the surface $f_i(S)$  equals $i$ times the maximum curvature within the $i$-scaled copy of $S$ within $f(S)$, which in the limit becomes arbitrarily close to $i$ times the curvature at $p$ in $f(S)$ and therefore has limiting value zero. The transformations $f_i$ preserve the tangent plane to the surface and directions within the tangent plane.

For each point $q\in S$, $f_i(q)$ can be obtained by the exponential map: follow a curve on $f_i(S)$ of length $|p-q|$, starting from $p$, in the direction given by the image of the tangent vector $q-p$. This length and direction are invariant through the motion, and as $i\to 0$ the curvature of this path approaches zero. Therefore, $f_i$ converges pointwise to a flat embedding $f_0(S)$, obtained by the exponential map on the tangent plane of $f(S)$ at $p$. Appending $f_0$ to our continuous sequence of smooth embeddings $f_i$ for $i\in(0,1]$ gives us a continuous sequence on $i\in [0,1]$, flat at $i=0$ and equal to our starting embedding at $i=1$, which therefore shows that these two embeddings are reconfigurable to each other.
\end{proof}

This proof uses compactness to ensure that the limit point of the spiral shrinking motion is asymptotically flat and that sufficiently small copies of the entire domain fit into any neighborhood of that point. Non-compact surfaces can have self-similar embeddings, smooth everywhere except the limit point, that are invariant under shrinking and re-expansion. For polygonal embeddings we handle the same issue of avoiding self-similar embeddings differently, using the requirement that these embeddings have finitely many connected linear pieces.

\begin{theorem}
\label{thm:unfold-polygonal}
Every polygonal embedding without interior vertices of a bounded spiral-shaped domain has a continuous motion, through polygonal embeddings without interior vertices, to a flat embedding.
\end{theorem}

\begin{proof}
The same idea as above comes close to working: compose the inward-spiraling shrinking motion of the domain, the initial embedding $f$, an expansion of $\mathbb{R}^3$ centered at the limit point of the motion, and a rotation of $\mathbb{R}^3$ that cancels any spinning motion the inward-spiraling motion might have. However, the limit point $p$ of the inward-spiraling shrinking motion might be a point on a fold line of the polygonal embedding, or worse, it might be a boundary point of $S$ where multiple fold lines meet. In this case, there is not a unique tangent plane of $f(S)$ at $p$, and when the same composition can be made to have a limit, this limit will be folded at $p$ in the same way as it was in $f(p)$ rather than being flat.

To address these issues, when $p$ is a folding point of the embedding $f(S)$, we choose one of the polygonal faces incident to $p$ to determine the tangent plane for the previous construction. Then as above we compose the inward-spiraling shrinking motion of the domain, the initial embedding $f$, an expansion of $\mathbb{R}^3$ centered at the limit point of the motion, and a rotation of $\mathbb{R}^3$ that cancels any spinning motion of this tangent plane. The resulting composition defines a continuous motion over the half-open interval of parameter values $(0,1]$, which can be extended with a well-defined limit at $0$, a polygonal folded state that has the same fold lines and fold angles as $f(S)$ at $p$ and is flat everywhere else.

We distinguish three cases:
\begin{itemize}
\item If $p$ is not a folding point of the embedding $f(S)$, then the previous proof applies directly.
\item If $p$ belongs to a single fold line of the embedding $f(S)$, then we can concatenate two continuous motions. First we perform the motion from the previous proof that transforms $f(S)$ into a folded state that contains this fold line. Because this folded state has only one fold line, it cannot self-intersect, and forms a valid polygonal embedding. Next, we perform an additional continuous motion that linearly changes the angle of this fold from its initial value to $\pi$ (unfolded into a flat angle). Again, none of the intermediate states of this second continuous motion can self-intersect.
\item In the remaining case, $p$ is a boundary point of $S$ that belongs to multiple fold lines within $S$. For a line from $p$ to intersect $S$ in a line segment, it must necessarily be the case that the inward-spiraling shrinking motion is actually the linear shrinking motion of a star-shaped domain. The \emph{link} of this folded state at $p$, the intersection of $f_0(S)$ with a small sphere centered at $p$, is a polygonal chain consisting of arcs of great circles. Because it remains invariant throughout the motion, it does not self-intersect. Any two distinct points of $S$ at distance $d$ from $p$ lie on distinct points of a scaled copy of this link, on a sphere of radius $d$, and for this reason cannot coincide. Therefore, the folded state at $f_0$ is again a valid polygonal embedding.

By known results on the spherical carpenter's rule problem, there exists a continuous motion of the link, as a polygonal chain of fixed-length great-circle arcs on the sphere, from its folded state to a completely flat state. This motion induces a continuous motion of $f_0$ to a flat-folded state~\cite{StrWhi-JCDCG-05}. Performing the continuous motion from $f(S)$ to $f_0(S)$, and then using this carpenter's rule solution on the resulting single-vertex surface, produces a combined continous motion from $f(S)$ to a flat state.\qedhere
\end{itemize}
\end{proof}

To apply the same method to flat foldings without interior vertices, we cannot use the carpenter's rule problem, as 1d flat-folded polygonal chains with fixed edge lengths cannot change continuously. Instead, we use a one-dimensional version of \cref{thm:unfold-polygonal}:

\begin{lemma}
\label{lem:1d-carpenter}
Let $P$ be line segment, folded flat by a piecewise-isometry $f:P\to R$ with a finite number of fold points and a consistent above-below relation for points with the same image. Then there exists a continuous motion of flat foldings of $P$ that transforms folding $f$ into an unfolded state.
\end{lemma}

\begin{proof}
Chose any point $p$ of $P$ that is not a fold point, and form a continuous family of one-dimensional folded states of $P$, $f_i(P)$ for $i\in(0,1]$, by scaling $P$ by a factor of $i$ centered at $p$, applying $f$, and scaling the result by a factor of $1/i$ centered at $f(p)$. When $i$ becomes less than the distance from $p$ to the nearest fold, the result will be an unfolding of $P$, so this provides a continuous transformation from $f$ to an unfolding.
\end{proof}

\begin{theorem}
Every flat folding without interior vertices of a bounded spiral-shaped domain has a continuous unfolding motion through foldings of the same type.
\end{theorem}

\begin{proof}
The proof follows the same outline as \cref{thm:unfold-polygonal}, combining shrinking and unshrinking to reach a folded state in which all folds pass through the center  $p$ of the spiral shrinking transformations. If the result has a single fold through $p$ we roll this fold to the boundary of $S$ rather than changing its fold angle. If $p$ lies on the boundary of the domain and belongs to multiple fold lines, we apply \cref{lem:1d-carpenter} in place of the spherical carpenter's rule.
\end{proof}

\section{An instructive example}
\begin{figure}[t]
\centering\includegraphics[width=0.8\columnwidth]{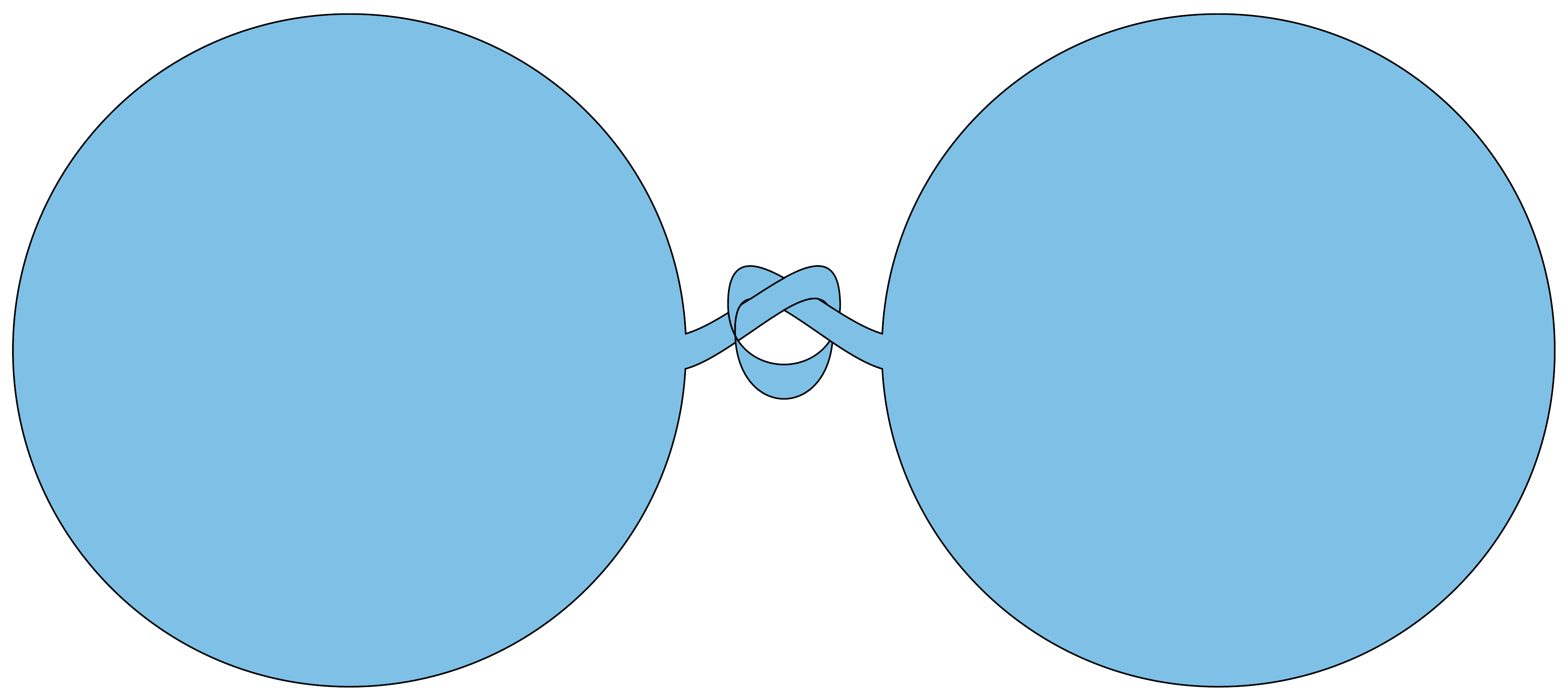}
\caption{Two disks connected by a knotted band. This surface can be flattened by a continuous motion of smooth embeddings.}
\label{fig:knotted-dumbbell}
\end{figure}

\cref{fig:knotted-dumbbell} shows two large disks connected by a short thin band, embedded with the disks spread flat and the band tied into an open overhead knot. If the disks could be crumpled, it would be easy to untie and flatten, by crumpling the disks into small enough balls that they could be passed through the knot, and then uncrumpling. However, our model of smooth surface embeddings does not allow crumpling. In every smooth embedding, the center point of each disk lies in a flat subset of the disk with large diameter: either a diameter of the disk, or a triangular subset of the disk with the vertices of this triangle on the boundary of the disk.

Rolling up either disk around a diameter makes that diameter act like a rigid line segment, but gives the rest of the disk a smaller overall shape. Rolling both disks in this way can produce an embedding like the locked polygonal chain with long ``knitting needles'' at its ends from Figure~1 of Biedl et al.~\cite{BieDemDem-DCG-01}. Any other rigid configuration for the two disks would be similarly locked. However, this surface is unlocked! It can be unfolded through the following sequence of transformations:
\begin{itemize}
\item Let $D$ be a diameter of one of the two disks, touching the boundary of the disk at its attachment point $p$ with the knotted band. Roll up the disk around $D$, starting at one of the points of the disk that is farthest from $D$, and leaving the semicircle opposite that point exposed on the outside of the roll.
\item Poke $p$ into the hole made by the knotted band, so that if the rolled-up disk around $D$ were not rigid, it could pass through the hole and untie the knot. However, because $D$ is made rigid by the bending of the embedding as it rolls around $D$, only the very end of diameter $D$ near $p$ can pass into the hole.
\item Continuously spin the parallel family of bend lines on the rolled-up disk, so that it rolls up around a different diameter than $D$. Choose the direction of spin that causes $p$ to travel along the exposed semicircle along the rolled-up disk. As it does so, this will allow more of the diameter of the rolled-up disk to poke through the hole in the knot.
\item When the bend lines have spun on the disk through an angle of $\pi$, causing $p$ to reach the other end of the exposed semicircle, the rolled-up diameter will have traveled all of the way through the hole in the knot, which will become unknotted.
\item Unroll the disk so that it lies flat with the rest of the surface.
\end{itemize}

Thus, although bend lines of a smoothly embedded surface are rigid, the underlying pattern of these lines on the surface can change continuously, complicating the search for a proof that a surface is locked.

\section{A locked surface}

\begin{figure}[t]
\centering\includegraphics[width=0.8\columnwidth]{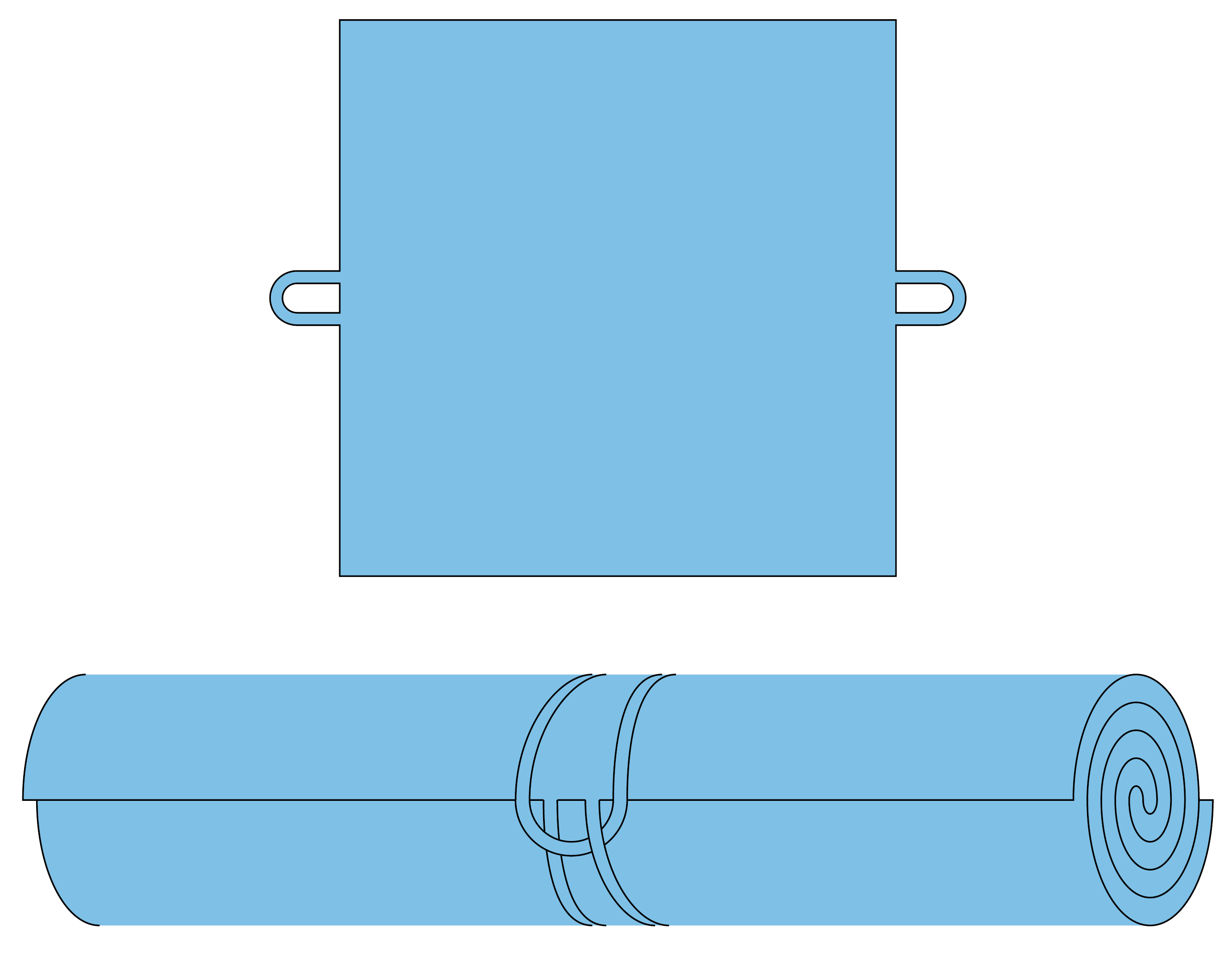}
\caption{A flat surface with two holes and a configuration that cannot be flattened. Intuitively, the two interlocked loops prevent the rolled center region from unrolling, the bend lines of the roll make it act like a rigid rod, and the length of this rod prevents the loops from being pulled around its ends. For the way the two loops on the hidden side of the roll interlock, see \cref{fig:borromean-cutaway}.}
\label{fig:tied-roll}
\end{figure}

\cref{fig:tied-roll} depicts a smoothly-embedded unit square with two small loops on midpoints of opposite sides, wrapped into a spiral roll  with the loops interlocked. \cref{fig:borromean-cutaway} provides a cutaway view showing how the loops interlock.

\begin{figure}[t]
\centering\includegraphics[width=0.8\columnwidth]{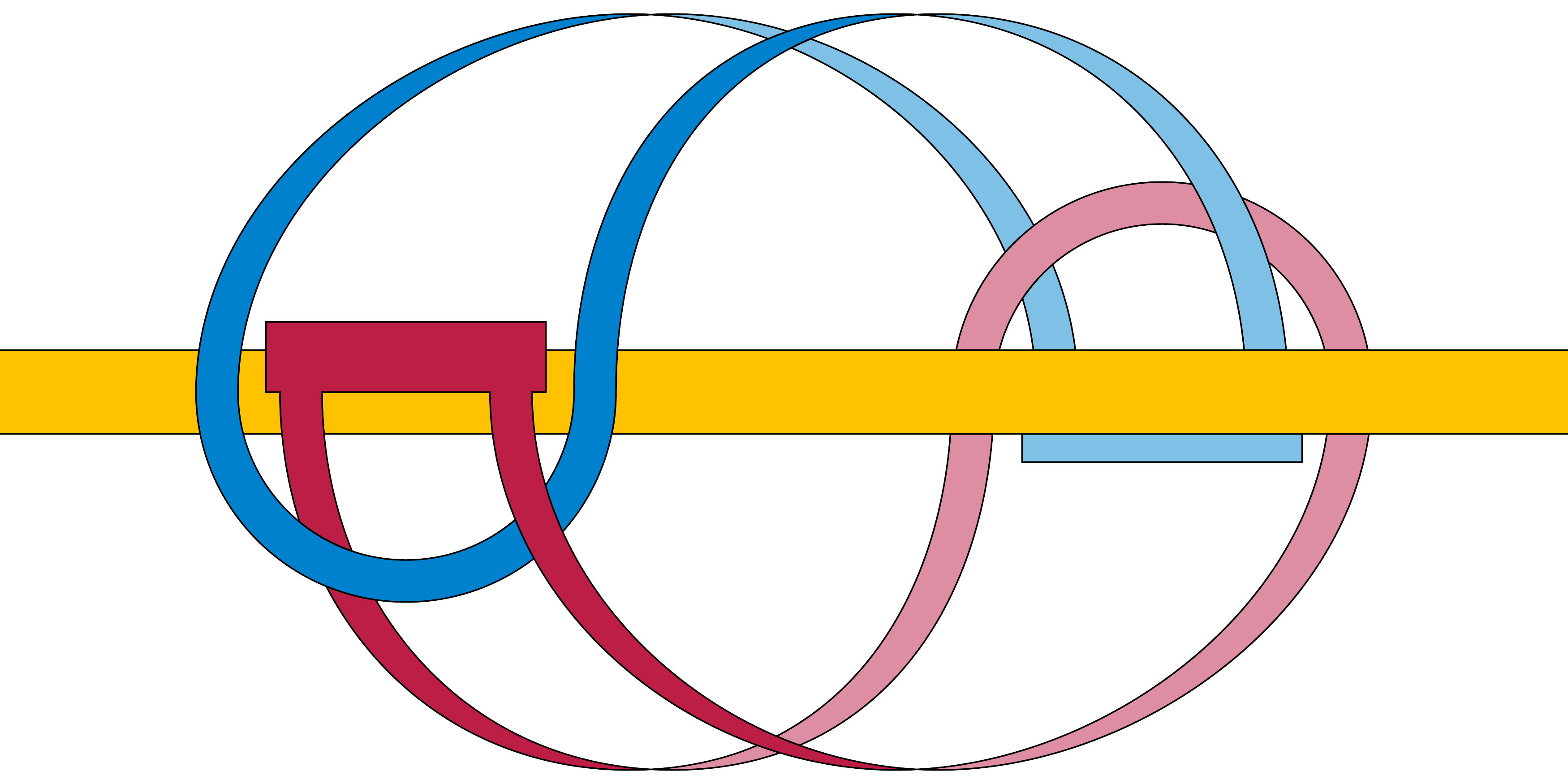}
\caption{Cutaway view of the two loops and centerline of the rolled-up part of the surface of \cref{fig:tied-roll}, showing their topological equivalence to the Borromean rings}
\label{fig:borromean-cutaway}
\end{figure}

This rolled-up and interlocked surface is topologically equivalent (ambient isotopic) to a flattened surface.
As can be seen in \cref{fig:borromean-cutaway}, the interlocking pattern is that of the Borromean rings, three unknotted loops in space that cannot be separated from each other, in which any two of the loops become unlinked if the third loop is removed. In the figure, the role of one of these three loops is taken by the center of the spiral roll, which does not actually form a loop, so this embedding is not topologically linked. From the arrangement of the loops in \cref{fig:borromean-cutaway}, it is possible to unlink it by pulling the right half of the red band to the right and up, around the right end of the yellow spiral center, passing this part of the red band around the right half of the blue band, and then passing the same part of the red band back to the right and down around the right end of the yellow spiral center. This sequence of motions reverses the front-back order of the two red-blue crossings on the right half of the red band, after which it is straightforward to flatten the whole surface. We have also verified that this configuration is topologically unlocked using a physical model of two rubber bands attached to a pen.

As a smoothly-embedded surface, \cref{fig:tied-roll} is locked: it cannot be flattened while preserving its geometry. To prove this, we use three interlocked properties of the embedding that, like the Borromean rings, are interlinked: as long as any two of the properties remain valid, the third one must remain valid as well, so none of the three properties can be the first to break in any continuous motion of the embedding. The first two properties are parameterized by a parameter $\varepsilon$, which we will be able to make arbitrarily small by making the length of the loops sufficiently small:

\begin{enumerate}[$A$.]
\item The two loops are bounded within distance $\varepsilon$ of each other.
\item The nearest bend line on the square to its center point has distance $\le\varepsilon$ to the center point, and  crosses the top and bottom sides of the square at distance $\ge 1/4-\varepsilon$ from its left and right sides.
\item The nearest bend line on the square to its center point can be completed into a loop by a curve, on a sphere with it as diameter, forming Borromean rings with the other two loops.
\end{enumerate}

\begin{figure}[t]
\centering\includegraphics[width=0.45\columnwidth]{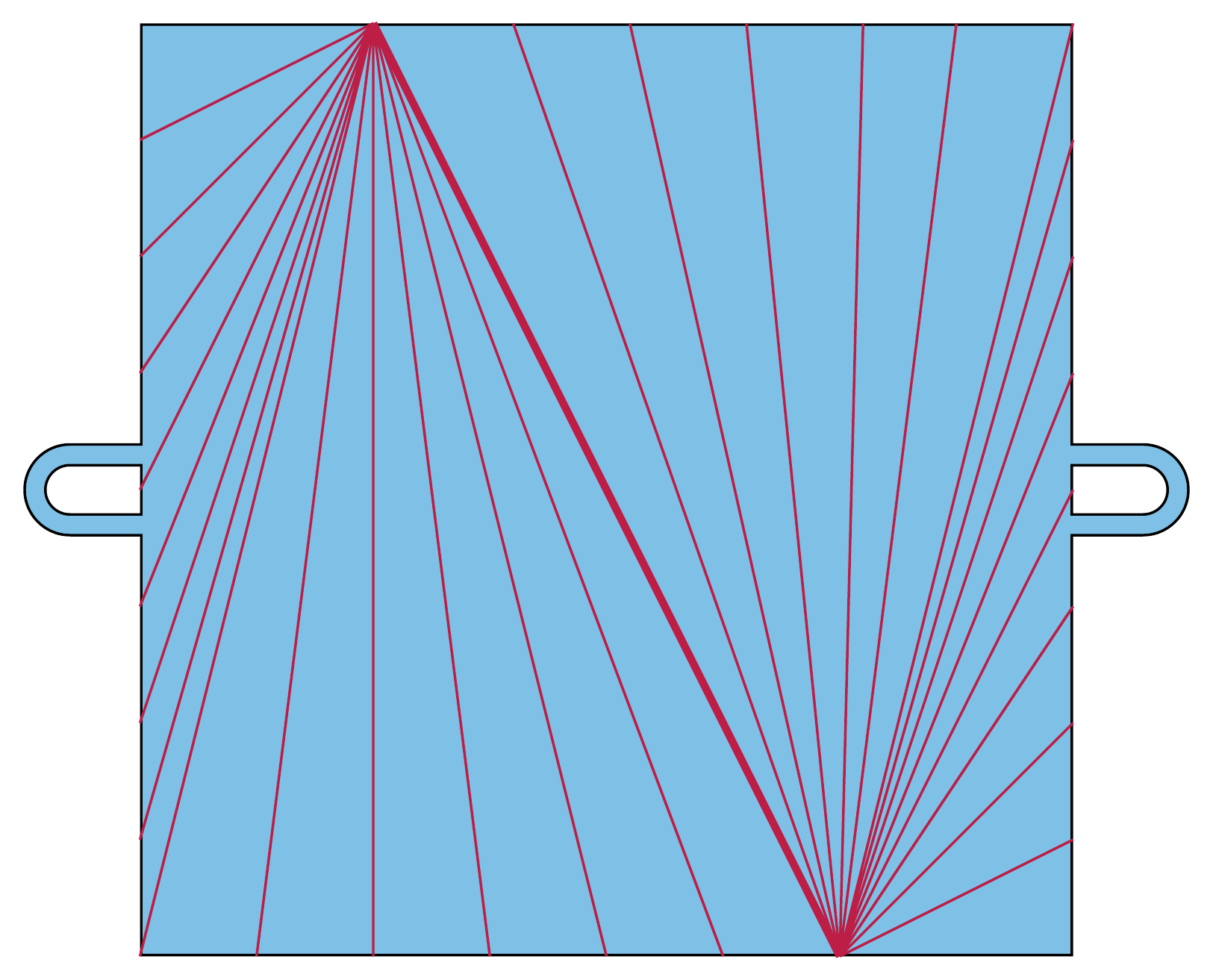}
\caption{Bend lines for a deformed version of \cref{fig:tied-roll} in which the (heavier) bend line through the center of the square crosses the top and bottom sides of the square at distance $1/4$ from the sides. For these bend lines, the two loops still meet near the center of the rolled surface.}
\label{fig:bent-roll}
\end{figure}

The reason for the $1/4$ in property $B$ is that it is possible to deform \cref{fig:tied-roll}, keeping the two loops interlocked in the center of the roll, so that the center bend line crosses the top and bottom sides of the square at distance $1/4$, as shown in \cref{fig:bent-roll}. However, as we will prove, it is not possible to move this center bend line significantly farther from vertical.

\begin{lemma}[$A\wedge B\Rightarrow C$]
\label{lem:AB2C}
For all sufficiently small loop lengths $\delta$ and all sufficiently small $\varepsilon$, continuous motion through states where $A(\varepsilon)$ and $B(\varepsilon)$ hold, starting from a state where $C$ holds, cannot reach a state where $C$ does not hold.
\end{lemma}

\begin{proof}
Let $b$ be the bend line of property~$B$. Then (\cref{fig:center-bend}) the attachment point $p$ of the right loop must lie within the intersection of two spheres centered at the endpoints of $b$, with radii equal to the distances from $p$ to those endpoints, because the embeddings we consider are not allowed to increase distances. Similarly, the attachment point of the left loop must lie within the intersection of another two spheres. By property $A$, these attachment points must be near each other, and until $C$ stops holding, they must also lie near line $b$. This limits their nearby locations to points of line $b$ that are far from its endpoints on the square, so they cannot reach the sphere through the endpoints of $b$ on which the connecting curve of property~$C$ lies.

As the smooth embedding deforms continuously, the bend line nearest the center point can change, but (as long as $B$ continues to hold) only by small amounts, and the connecting curve can be changed by similar small amounts to maintain property $C$. As the other two loops cannot reach the sphere containing this curve, they cannot cross this curve (even though it does not form a physical obstacle to them) and cannot change the knotted topology that it forms with them.
\end{proof}

\begin{lemma}[$A\wedge C\Rightarrow B$]
\label{lem:AC2B}
For all sufficiently small $\varepsilon$ there exists a $\delta$ such that, for loop lengths less than $\delta$,
continuous motion through states where $A(2\varepsilon)$ and $C$ hold, starting from a state where $B(\varepsilon)$ holds,
cannot reach a state where $B(\varepsilon)$ does not hold.
\end{lemma}

\begin{proof}
Let $b$ be the bend line nearest the center of the square. While $A$ and $C$ hold, $b$ must pass through the Borromean link formed by it and the two loops, and so (if the center point itself does not lie on a bend line) it must lie on a flattened part of the surface whose width is at most proportional to the loop length. Therefore, $b$ is close to the center point, as part of property $B$ demands. For a bend line that is close to the center point, the same reasoning used in \cref{fig:center-bend} and \cref{lem:AB2C} shows that it must be at distance at least $1/4-\varepsilon$ from the left and right sides, for otherwise the two intervals on this bend line where the left and right loop attachment points must be near would not intersect.
\end{proof}

\begin{lemma}[$B\wedge C\Rightarrow A$]
\label{lem:BC2A}
For all sufficiently small $\varepsilon$ there exists a $\delta$ such that, for loop lengths less than $\delta$,
continuous motion through states where $B(2\varepsilon)$ and $C$ hold, starting from a state where $A(\varepsilon)$ holds,
cannot reach a state where $A(\varepsilon)$ does not hold.
\end{lemma}

\begin{figure}[t]
\centering\includegraphics[width=0.8\columnwidth]{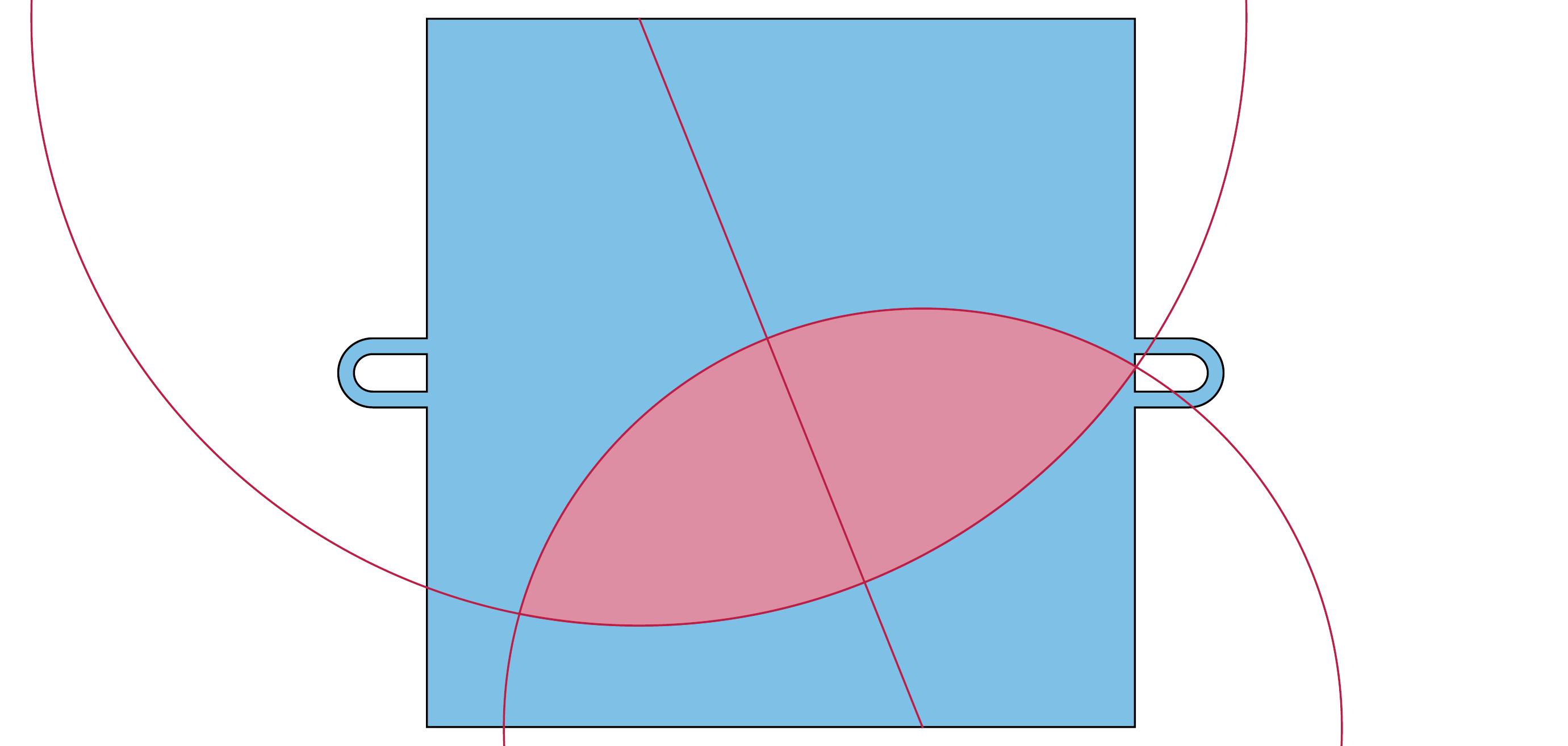}
\caption{For a bend line through the square's center, the attachment point of the right loop must be in the shaded intersection of circles centered at the bend endpoints, to avoid having greater distance to those endpoints than in the flattened surface.}
\label{fig:center-bend}
\end{figure}

\begin{proof}
By assumption $C$, the two loops of the surface and a third loop formed by bend line $b$ form Borromean rings, and by assumption $B$, line $b$ is near vertical and near the square's center, forcing the attachment points of the two loops to be far from the endpoints of~$b$.

Each loop of the surface has small diameter, so if the two loops could be far from each other it would be possible to make two small spheres (of radius $\varepsilon$), one containing each loop. The loop containing $b$ lies on a straight line within each of these two spheres. Although Borromean rings can have two loops in two disjoint spheres~\cite{Thi-PJM-91} (\cref{fig:borromean-spheres}), the third Borromean ring cannot pass through either sphere as a straight line.
If it could, we could deform a loop within its sphere to a circle (as it is not pairwise linked with the line passing through the sphere) and span it by a disk not crossed by the other two loops, impossible for the Borromean rings. This contradiction shows that the two loops cannot be separated by a distance larger than $\varepsilon$, as stated in property~$A$.
\end{proof}

These properties together prove that the surface of \cref{fig:tied-roll} is locked: for versions of this surface with short enough loops, it is impossible to deform it as a smoothly embedded surface to its flattened state. The same is true for the same reasons for approximations to this surface by polyhedral embeddings or flat foldings without interior vertices. As a result, we have the following theorem:

\begin{figure}[t]
\centering\includegraphics[width=0.7\columnwidth]{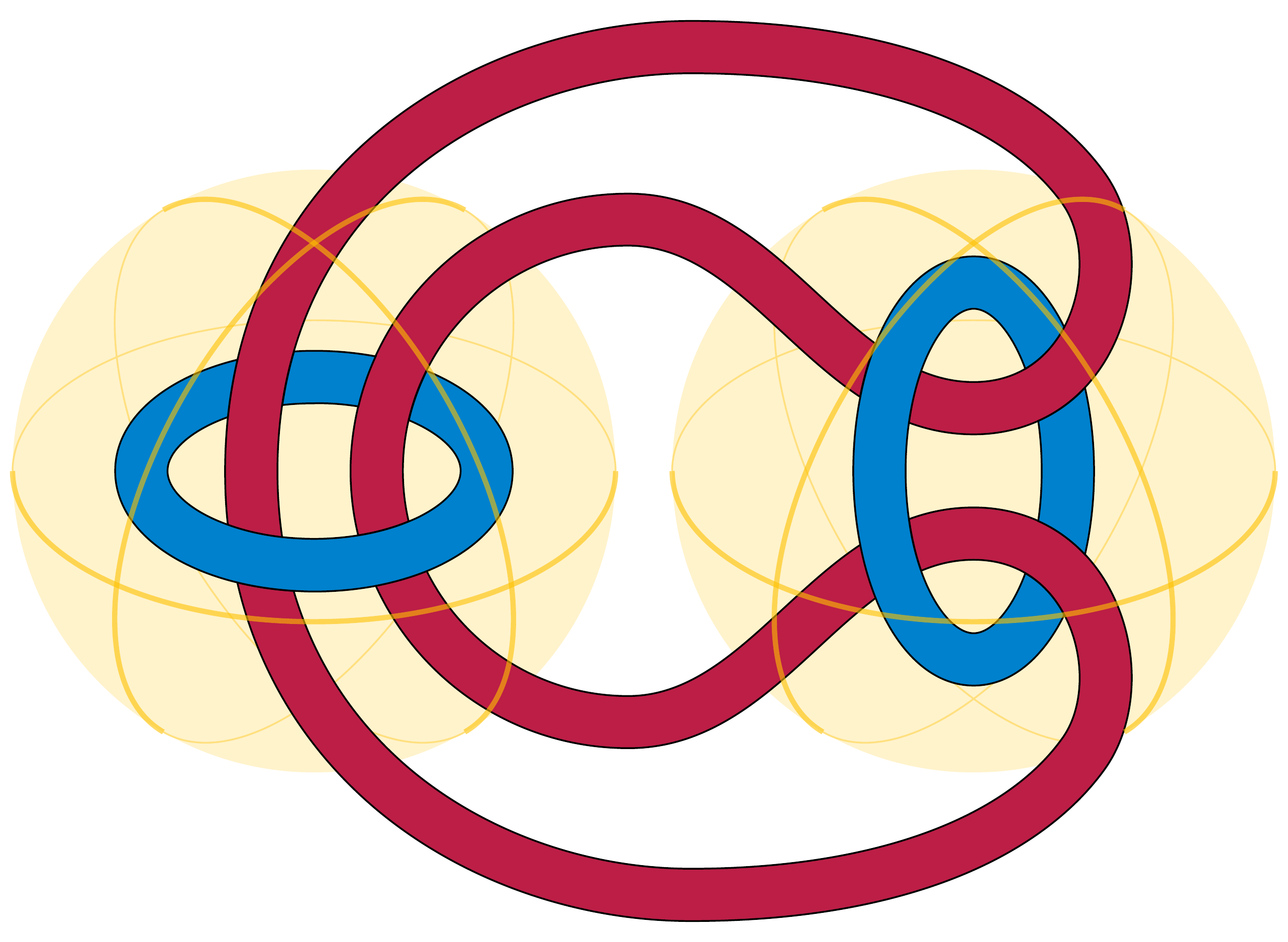}
\caption{Borromean rings with links in separate spheres}
\label{fig:borromean-spheres}
\end{figure}

\begin{theorem}
For smooth embeddings, polyhedral embeddings without interior vertices, and flat foldings without interior vertices, there exist flat surfaces with the topology of a disk with two holes that are ambient isotopic to their flattened form but cannot reach that form by a continuous sequence of folded states staying within the same class of folded states.
\end{theorem}

\begin{proof}
Choose a sufficiently small $\delta>0$ and $\varepsilon>0$ for the lemmas above. The surface of \cref{fig:tied-roll}, in the configuration of the figure, has properties $A(\varepsilon)$, $B(\varepsilon)$, and $C$. As it continuously moves, all three properties remain true; none can be the first to fail, because at the instant it failed, the weaker properties $A(2\varepsilon)$ and $B(2\varepsilon)$ would still be true, which by the lemmas would imply all three of $C$, $A(\varepsilon)$, and $B(\varepsilon)$. However, the flattened configuration does not obey the three properties, so it cannot be reached.
\end{proof}

\section{Conclusions and open problems}

We have shown that, for smooth embeddings, polyhedral embeddings without interior vertices, and flat foldings without interior vertices, every spiral-shaped domain can be flattened. However, there exist more complex planar shapes whose configuration spaces are disconnected: they have locked embeddings that, although topologically equivalent, cannot be flattened.

Is every topological disk flattenable? Can the method that we used to flatten \cref{fig:knotted-dumbbell} be generalized to other disks? What about surfaces with a single hole? Additionally, we have only investigated the existence of flattenings, but not their algorithmic complexity. For polyhedral embeddings and flat foldings, how hard is it to determine whether a continuous flattening exists? What about for smooth embeddings represented as a piecewise cylindrical and conical surface?

\bibliographystyle{plainurl}
\bibliography{developable}

\begin{thebibliography}{10}

\bibitem{AbeCanDem-JoCG-16}
Zachary Abel, Jason Cantarella, Erik~D. Demaine, David Eppstein, Thomas~C.
  Hull, Jason~S. Ku, Robert~J. Lang, and Tomohiro Tachi.
\newblock {Rigid origami vertices: conditions and forcing sets}.
\newblock {\em Journal of Computational Geometry}, 7(1):171{--}184, 2016.
\newblock \href {http://dx.doi.org/10.20382/jocg.v7i1a9}
  {\path{doi:10.20382/jocg.v7i1a9}}.

\bibitem{AhaEliSho-JMAA-03}
Dov Aharonov, Mark Elin, and David Shoikhet.
\newblock {Spiral-like functions with respect to a boundary point}.
\newblock {\em Journal of Mathematical Analysis and Applications},
  280(1):17{--}29, 2003.
\newblock \href {http://dx.doi.org/10.1016/S0022-247X(02)00615-7}
  {\path{doi:10.1016/S0022-247X(02)00615-7}}.

\bibitem{BalChaDem-WADS-09}
Brad Ballinger, David Charlton, Erik~D. Demaine, Martin~L. Demaine, John
  Iacono, Ching-Hao Liu, and Sheung-Hung Poon.
\newblock {Minimal locked trees}.
\newblock In Frank K. H.~A. Dehne, Marina~L. Gavrilova, J{\"o}rg{-}R{\"u}diger
  Sack, and Csaba~D. T{\'o}th, editors, {\em Algorithms and Data Structures,
  11th International Symposium, WADS 2009, Banff, Canada, August 21-23, 2009.
  Proceedings}, volume 5664 of {\em Lecture Notes in Computer Science}, pages
  61{--}73. Springer, 2009.
\newblock \href {http://dx.doi.org/10.1007/978-3-642-03367-4_6}
  {\path{doi:10.1007/978-3-642-03367-4_6}}.

\bibitem{BieDemDem-DCG-01}
Therese Biedl, Erik~D. Demaine, Martin~L. Demaine, Sylvain Lazard, Anna Lubiw,
  Joseph O'Rourke, Mark Overmars, Steve Robbins, Ileana Streinu, Godfried
  Toussaint, and Sue Whitesides.
\newblock {Locked and unlocked polygonal chains in three dimensions}.
\newblock {\em Discrete {\&} Computational Geometry}, 26(3):269{--}281, 2001.
\newblock \href {http://dx.doi.org/10.1007/s00454-001-0038-7}
  {\path{doi:10.1007/s00454-001-0038-7}}.

\bibitem{BieDemDem-DAM-02}
Therese Biedl, Erik~D. Demaine, Martin~L. Demaine, Sylvain Lazard, Anna Lubiw,
  Joseph O'Rourke, Steve Robbins, Ileana Streinu, Godfried Toussaint, and Sue
  Whitesides.
\newblock {A note on reconfiguring tree linkages: trees can lock}.
\newblock {\em Discrete Applied Mathematics}, 117(1-3):293{--}297, 2002.
\newblock \href {http://dx.doi.org/10.1016/S0166-218X(01)00229-3}
  {\path{doi:10.1016/S0166-218X(01)00229-3}}.

\bibitem{BieLubSun-CG-05}
Therese Biedl, Anna Lubiw, and Julie Sun.
\newblock {When can a net fold to a polyhedron?}
\newblock {\em Computational Geometry}, 31(3):207{--}218, 2005.
\newblock \href {http://dx.doi.org/10.1016/j.comgeo.2004.12.004}
  {\path{doi:10.1016/j.comgeo.2004.12.004}}.

\bibitem{ConDemDem-DCG-10}
Robert Connelly, Erik~D. Demaine, Martin~L. Demaine, S{\'a}ndor~P. Fekete,
  Stefan Langerman, Joseph S.~B. Mitchell, Ares Rib{\'o}, and G{\"u}nter Rote.
\newblock {Locked and unlocked chains of planar shapes}.
\newblock {\em Discrete {\&} Computational Geometry}, 44(2):439{--}462, 2010.
\newblock \href {http://dx.doi.org/10.1007/s00454-010-9262-3}
  {\path{doi:10.1007/s00454-010-9262-3}}.

\bibitem{ConDemRot-DCG-03}
Robert Connelly, Erik~D. Demaine, and G{\"u}nter Rote.
\newblock {Straightening polygonal arcs and convexifying polygonal cycles}.
\newblock {\em Discrete {\&} Computational Geometry}, 30(2):205{--}239, 2003.
\newblock \href {http://dx.doi.org/10.1007/s00454-003-0006-7}
  {\path{doi:10.1007/s00454-003-0006-7}}.

\bibitem{DemDemHar-GC-11}
Erik~D. Demaine, Martin~L. Demaine, Vi~Hart, John Iacono, Stefan Langerman, and
  Joseph O'Rourke.
\newblock {Continuous blooming of convex polyhedra}.
\newblock {\em Graphs and Combinatorics}, 27(3):363{--}376, 2011.
\newblock \href {http://dx.doi.org/10.1007/s00373-011-1024-3}
  {\path{doi:10.1007/s00373-011-1024-3}}.

\bibitem{DemDevMit-CCCG-04}
Erik~D. Demaine, Satyan~L. Devadoss, Joseph S.~B. Mitchell, and Joseph
  O'Rourke.
\newblock {Continuous foldability of polygonal paper}.
\newblock In {\em Proceedings of the 16th Canadian Conference on Computational
  Geometry, CCCG'04, Concordia University, Montr{\'e}al, Qu{\'e}bec, Canada,
  August 9-11, 2004}, pages 64{--}67, 2004.
\newblock URL: \url{https://www.cccg.ca/proceedings/2004/55.pdf}.

\bibitem{DemLanOro-CG-03}
Erik~D. Demaine, Stefan Langerman, Joseph O'Rourke, and Jack Snoeyink.
\newblock {Interlocked open and closed linkages with few joints}.
\newblock {\em Computational Geometry}, 26(1):37{--}45, 2003.
\newblock \href {http://dx.doi.org/10.1016/S0925-7721(02)00171-2}
  {\path{doi:10.1016/S0925-7721(02)00171-2}}.

\bibitem{DemMit-CCCG-01}
Erik~D. Demaine and Joseph S.~B. Mitchell.
\newblock {Reaching folded states of a rectangular piece of paper}.
\newblock In {\em Proceedings of the 13th Canadian Conference on Computational
  Geometry, University of Waterloo, Ontario, Canada, August 13-15, 2001}, pages
  73{--}75, 2001.
\newblock URL:
  \url{https://erikdemaine.org/papers/PaperReachability_CCCG2001/}.

\bibitem{Epp-JoCG-19}
David Eppstein.
\newblock {Realization and connectivity of the graphs of origami flat
  foldings}.
\newblock {\em Journal of Computational Geometry}, 10(1):257{--}280, 2019.
\newblock \href {http://dx.doi.org/10.20382/jocg.v10i1a10}
  {\path{doi:10.20382/jocg.v10i1a10}}.

\bibitem{FucTab-MO-07}
Dmitry Fuchs and Serge Tabachnikov.
\newblock {Lecture 14: Paper M{\"o}bius band}.
\newblock In {\em Mathematical Omnibus: Thirty Lectures on Classic
  Mathematics}, pages 199{--}206. American Mathematical Society, Providence,
  Rhode Island, 2007.
\newblock \href {http://dx.doi.org/10.1090/mbk/046}
  {\path{doi:10.1090/mbk/046}}.

\bibitem{HaoKimLie-SCF-18}
Yue Hao, Yun hyeong Kim, and Jyh-Ming Lien.
\newblock {Synthesis of fast and collision-free folding of polyhedral nets}.
\newblock In {\em Proceedings of the 2nd ACM Symposium on Computational
  Fabrication}, June 2018.
\newblock \href {http://dx.doi.org/10.1145/3213512.3213517}
  {\path{doi:10.1145/3213512.3213517}}.

\bibitem{HopJosWhi-SICOMP-84}
John Hopcroft, Deborah Joseph, and Sue Whitesides.
\newblock {Movement problems for 2-dimensional linkages}.
\newblock {\em SIAM Journal on Computing}, 13(3):610{--}629, 1984.
\newblock \href {http://dx.doi.org/10.1137/0213038}
  {\path{doi:10.1137/0213038}}.

\bibitem{LeePre-JACM-79}
D.~T. Lee and F.~P. Preparata.
\newblock {An optimal algorithm for finding the kernel of a polygon}.
\newblock {\em Journal of the ACM}, 26(3):415{--}421, 1979.
\newblock \href {http://dx.doi.org/10.1145/322139.322142}
  {\path{doi:10.1145/322139.322142}}.

\bibitem{MilPak-DCG-08}
Ezra Miller and Igor Pak.
\newblock {Metric combinatorics of convex polyhedra: Cut loci and
  nonoverlapping unfoldings}.
\newblock {\em Discrete {\&} Computational Geometry}, 39(1-3):339{--}388, 2008.
\newblock \href {http://dx.doi.org/10.1007/s00454-008-9052-3}
  {\path{doi:10.1007/s00454-008-9052-3}}.

\bibitem{PanStr-CG-10}
Gaiane Panina and Ileana Streinu.
\newblock {Flattening single-vertex origami: The non-expansive case}.
\newblock {\em Computational Geometry}, 43(8):678{--}687, 2010.
\newblock \href {http://arxiv.org/abs/1003.3490} {\path{arXiv:1003.3490}},
  \href {http://dx.doi.org/10.1016/j.comgeo.2010.04.002}
  {\path{doi:10.1016/j.comgeo.2010.04.002}}.

\bibitem{Par-TAMS-09}
John Pardon.
\newblock {On the unfolding of simple closed curves}.
\newblock {\em Transactions of the American Mathematical Society},
  361(4):1749{--}1764, 2009.
\newblock \href {http://dx.doi.org/10.1090/S0002-9947-08-04781-8}
  {\path{doi:10.1090/S0002-9947-08-04781-8}}.

\bibitem{SonAma-TRA-04}
Guang Song and Nancy~M. Amato.
\newblock {A motion-planning approach to folding: From paper craft to protein
  folding}.
\newblock {\em IEEE Transactions on Robotics and Automation}, 20(1):60{--}71,
  February 2004.
\newblock \href {http://dx.doi.org/10.1109/tra.2003.820926}
  {\path{doi:10.1109/tra.2003.820926}}.

\bibitem{Str-FOCS-00}
Ileana Streinu.
\newblock {A combinatorial approach to planar non-colliding robot arm motion
  planning}.
\newblock In {\em Proceedings of the 41st IEEE Symposium on Foundations of
  Computer Science}, pages 443{--}453, 2000.
\newblock \href {http://dx.doi.org/10.1109/SFCS.2000.892132}
  {\path{doi:10.1109/SFCS.2000.892132}}.

\bibitem{StrWhi-JCDCG-05}
Ileana Streinu and Walter Whiteley.
\newblock {Single-vertex origami and spherical expansive motions}.
\newblock In {\em Discrete and Computational Geometry: Japanese Conference,
  JCDCG 2004, Tokyo, Japan, October 8-11, 2004, Revised Selected Papers},
  volume 3742 of {\em Lecture Notes in Computer Science}, pages 161{--}173.
  Springer-Verlag, 2005.
\newblock \href {http://dx.doi.org/10.1007/11589440_17}
  {\path{doi:10.1007/11589440_17}}.

\bibitem{Thi-PJM-91}
Morwen~B. Thistlethwaite.
\newblock {On the algebraic part of an alternating link}.
\newblock {\em Pacific Journal of Mathematics}, 151(2):317{--}333, 1991.
\newblock \href {http://dx.doi.org/10.2140/pjm.1991.151.317}
  {\path{doi:10.2140/pjm.1991.151.317}}.

\bibitem{TurGooSen-PIMEC-15}
Nicholas Turner, Bill Goodwine, and Mihir Sen.
\newblock {A review of origami applications in mechanical engineering}.
\newblock {\em Proceedings of the Institution of Mechanical Engineers, Part C},
  230(14):2345{--}2362, August 2015.
\newblock \href {http://dx.doi.org/10.1177/0954406215597713}
  {\path{doi:10.1177/0954406215597713}}.

\bibitem{XiLie-IROS-15}
Zhonghua Xi and Jyh-Ming Lien.
\newblock {Continuous unfolding of polyhedra {--} a motion planning approach}.
\newblock In {\em 2015 IEEE/RSJ International Conference on Intelligent Robots
  and Systems (IROS)}. IEEE, September 2015.
\newblock \href {http://dx.doi.org/10.1109/iros.2015.7353828}
  {\path{doi:10.1109/iros.2015.7353828}}.

\end{thebibliography}

\end{document}